\documentclass[12pt]{amsart}

\usepackage{amsfonts,amsmath,amsthm}
\usepackage{latexsym}

\newtheorem{theorem}{Theorem}[section]
\newtheorem{corollary}[theorem]{Corollary}
\newtheorem{lemma}[theorem]{Lemma}

\theoremstyle{definition}

\newcommand{\dst}{\displaystyle}

\newcommand \tr {\mathrm{Tr}\:}

\newcommand{\Co}{\ensuremath{\mathbb{C}}}

\newcommand{\ac}{\ensuremath{\mathcal{A}}}
\newcommand{\bc}{\ensuremath{\mathcal{B}}}

\newcommand{\vb}{\ensuremath{\mathbf{b}}}

\newcommand{\vone}{\ensuremath{\mathbf{1}}}
\newcommand{\eb}{\ensuremath{\mathbf{e}}}
\newcommand{\fb}{\ensuremath{\mathbf{f}}}

\def \< {\langle}
\def \> {\rangle}

\newcommand{\abs}[1]{{\left|{#1}\right|}}
\newcommand{\scal}[1]{{\left\langle{#1}\right\rangle}}

\begin{document}

\title[A rigidity property of MUBs]{A rigidity property of complete systems of mutually unbiased bases}

\author[M. Matolcsi]{M\'at\'e Matolcsi}
\address{M.M.: Budapest University of Technology and Economics (BME),
H-1111, Egry J. u. 1, Budapest, Hungary (also at Alfr\'ed R\'enyi Institute of Mathematics,
Hungarian Academy of Sciences, H-1053, Realtanoda u 13-15, Budapest, Hungary)}
\email{matomate@renyi.hu}

\author[M. Weiner]{Mih\'aly Weiner}
\address{M.W.: Budapest University of Technology and Economics (BME), H-1111, Egry J. u. 1, Budapest, Hungary and MTA-BME Lend\"ulet Quantum Information Theory Research Group}
\email{mweiner@math.bme.hu}

\thanks{M. Matolcsi was supported by
NKFI grants K132097, K129335. M. Weiner was supported by the Bolyai Janos Fellowship of the Hungarian Academy of Sciences, the UNKP-20-5 New National Excellence Program of the Ministry for Innovation and Technology and by NKFI grants K132097, K124152 and KH129601.}

\begin{abstract}
Suppose that for some unit vectors $\vb_1,\ldots \vb_n$ in $\mathbb C^d$ we have that for any $j\neq k$ $\vb_j$ is either orthogonal to $\vb_k$ or $|\langle \vb_j,\vb_k\rangle|^2 = 1/d$ (i.e. $\vb_j$ and $\vb_k$ are unbiased). We prove that if $n=d(d+1)$, then these vectors necessarily form a complete system of mutually unbiased bases, that is, they can be arranged into $d+1$ orthonormal bases, all being mutually unbiased with respect to each other.
\end{abstract}

\maketitle

\bigskip

\section{Introduction}

The concept of mutually unbiased bases (MUBs) originates from quantum state tomography (\cite{ivanovic}), and appears also in several protocols in quantum  information theory (\cite{hal}). As such, the existence and explicit constructions of MUBs have been active areas of research in the past decades (see e.g. \cite{sum} for a recent comprehensive survey article).

\medskip

Recall that two orthonormal bases in $\Co^d$, $\ac=\{\eb_1,\ldots,\eb_d\}$ and $\bc=\{\fb_1,\ldots,\fb_d\}$ are called \emph{unbiased} if for every $1\leq j,k\leq d$, $\abs{\scal{\eb_j,\fb_k}}=\dst\frac{1}{\sqrt{d}}$. A collection $\bc_1,\ldots\bc_m$ of orthonormal bases is said to be (pairwise) \emph{mutually unbiased} if any two of them are unbiased. If the dimension $d$ is a prime-power, then the maximal number of MUBs is well-known to be $d+1$ (see e.g. \cite{ivanovic, WF,BBRV, KR}). It is also well-known that in any dimension $d$ the maximal number of MUBs is at most $d+1$ (see e.g. \cite{WF,belovs,mubfourier,sum}). For this reason, a set of $d+1$ mutually unbiased bases is commonly called a {\it complete system of MUBs}. However, for any $d$ which is not a prime-power, it is not known whether a complete system of MUBs exists (even for $d=6$, despite considerable efforts \cite{BBELTZ,boykin,config,arxiv}).

\medskip

In \cite[Theorem 8]{belovs} it is proved that unit vectors forming a complete system of MUBs, if they exist, must satisfy some extra algebraic relations. Furthermore, in \cite[Theorem 2.2]{mubfourier} the following result is proved: a collection of unit vectors in $\Co^d$, all of which are orthogonal or unbiased to a fixed orthonormal basis, can consist of at most $d^2$ vectors. These two results raise the following very general and natural question: given a set of $d(d+1)$ unit vectors in $\Co^d$ such that any two of them are {\it either} orthogonal {\it or} unbiased to each other, is it true that they necessarily form a complete system of MUBs?  In this paper we answer this question in the affirmative, which can be viewed as a certain rigidity property of complete systems of MUB's. This result is somewhat surprising, considering that as many as $(d-1)^2$ unit vectors in $\Co^d$ can be given such that they are pairwise unbiased to each other. Indeed, consider a SIC-POVM (which conjecturally exists in any dimension) in $\Co^{d-1}$, i.e. a collection of $(d-1)^2$ unit vectors in $\Co^{d-1}$ such that any pair has inner product with absolute value $\frac{1}{\sqrt{d}}$. Append each vector with a coordinate 0 in the $d$th coordinate, and you obtain a collection of $(d-1)^2$ unit vectors in $\Co^d$ which are pairwise unbiased to each other. One might expect that similar special constructions may yield $d(d+1)$ unit vectors in several different ways, such that they are all orthogonal {\it or} unbiased to each other, but Theorem \ref{main} tells us that this is not the case.

\section{From a set of vectors to a complete system of MUBs}

Suppose that $n=d(d+1)$ and $\mathbf b_1,\ldots \mathbf b_n \in \mathbb C^d$ is a collection of unit vectors such that any two of them is either orthogonal or unbiased to each other, that is $|\langle \vb_j, \vb_k\rangle|=0$ or $\frac{1}{\sqrt{d}}$ for any $j\ne k$. We will prove below (Theorem \ref{main}) that these vectors necessarily form a complete system of MUBs.

\medskip

Consider the simple graph $G=(V,E)$ with vertex set $V=\{\mathbf b_1,\ldots \mathbf b_n\}$ and edge set $E$ containing all (unordered) pairs of orthogonal vectors in $V$. In other words, we imagine that vectors $\vb_j$ and $\vb_k$ are connected by an edge if they are orthogonal to each other. Our aim is to prove that $G$ is a disjoint union of $d+1$ complete graphs, each containing $d$ vertices. This will prove that the vectors $\vb_j$ can be grouped into $d+1$ orthonormal bases, all being mutually unbiased to each other. We shall begin by considering the number of edges in $G$. Note that if the vectors in $V$ form $d+1$ mutually unbiased bases, then the number of orthogonality relations (i.e. the number of edges in $G$) should be $(d+1)\binom{d}{2}$.

\medskip

 The following is a well-known general fact, but we include it for the convenience of the reader.

\begin{lemma}
\label{lemma:rk}
Suppose $A$ is a self-adjoint matrix of rank $r={\rm rk}(A)$. Then $(\tr A)^2 \leq r\, \tr(A^2)$ with equality holding if and only if $A$ is a multiple of a projection.
\end{lemma}

\begin{proof}
We may assume that the rank $r>0$ (the case $r=0$ implies $A=0$, which is trivial).
Let $P$ be the orthogonal projection onto the range space of $A$. Then $PA=A$ and $\tr(P^2)=\tr(P)={\rm rk}(P)={\rm rk}(A)=r$. Using the Cauchy-Schwarz inequality $|\tr(X^*Y)|^2\leq \tr(X^*X)\tr(Y^*Y)$ we have
$$
(\tr A)^2=(\tr PA)^2\leq\tr(P^2)\,\tr(A^2)=  r \tr(A^2)
$$
with equality holding if and only if $A$ and $P$ are parallel; i.e.\! when $A$ is a multiple of $P$.
\end{proof}
\begin{corollary}
\label{corr:ort_from_above}
The graph $G$ has at most $(d+1)\binom{d}{2}$ edges.
\end{corollary}
\begin{proof}
We will denote the number of edges by $|E|$.
Consider the Gram matrix
$$
K := \left(\langle\mathbf b_j,\mathbf b_k\rangle\right)_{\{j,k\}}
$$
of the given vectors.
The rank of $K$ is the dimension of the subspace spanned by the vectors $\mathbf b_1,\ldots \mathbf b_n \in \mathbb C^d$ and hence ${\rm rk}(K)\leq d$. Since these vectors are of unit length, the diagonal elements of $K$ are all equal to $1$ and thus $\tr(K)=n=d(d+1)$. Moreover, as
$K$ is self-adjoint (actually: positive semidefinite),
$$
\tr(K^2)=\tr(K^*K)=\sum_{j,k}|K_{j,k}|^2
=\sum_{j,k}|\langle\mathbf b_j,\mathbf b_j\rangle|^2.
$$
In the above sum, we have $3$ kind of terms. First, the ones with $j=k$, of which we have $n=d(d+1)$ many. Second, the ones  corresponding to orthogonal pairs of vectors; of these we have $2|E|$ -- the factor of $2$ needed because we considered $G$ to be undirected. Finally, we have the ones corresponding to unbiased pairs of vectors; of these we have $2\left(\binom{n}{2} -|E|\right)$. So
$$
\tr(K^2)=n\cdot 1 + 2 |E|\cdot 0 + 2\left(\binom{n}{2} -|E|\right)\cdot \frac{1}{d} = n+\frac{n(n-1)}{d}-\frac{2 |E|}{d},
$$
and hence by the previous lemma
$$
n^2 \leq d \left(n+\frac{n(n-1)}{d}-\frac{2 |E|}{d}\right).
$$
Substituting $n=d(d+1)$ and rearranging we get $|E|\leq \frac{d(d^2-1)}{2}$, which is the claimed bound.
\end{proof}

To completely determine $|E|$, we also need to bound it from below.
This means bounding the number of non-orthogonal (i.e.\! unbiased) pairs from above.
More concretely, we need to show that using the vectors $\mathbf b_1,\ldots \mathbf b_n$, one can form at most $\binom{d+1}{2}d^2$ unbiased pairs; i.e.\! exactly as many as we would have if these vectors were to from a complete system of MUBs.

\medskip

To this end, for each $j\in\{1,\ldots n\}$ consider $Q_j:=|\mathbf b_j\rangle\langle \mathbf b_j|$, i.e.\! the orthogonal projection onto the one-dimensional subspace given by the vector $\mathbf b_j$, and let $X_j=Q_j-\frac{1}{d}I$. Elementary computation shows that the Hilbert-Schmidt inner products satisfy
\begin{equation}
\label{eq:<X_jX_k>}
\langle X_j,X_k\rangle_{HS}= \tr(X_j^* X_k)=
|\langle \mathbf b_j, \mathbf b_k\rangle|^2-\frac{1}{d},
\end{equation}
where $\langle\cdot,\cdot\rangle_{HS}$ denotes the usual Hilbert-Schmidt inner product on $M_d(\mathbb C)$. We shall now apply the estimate of Lemma \ref{lemma:rk} to the Gram matrix
$$
\tilde{K}:=\left(\langle X_j,X_k \rangle_{HS}\right)_{\{j,k\}}
$$
Note that $\tilde{K}$ has size $n\times n$.

\begin{lemma}
The graph $G$ has exactly $(d+1)\binom{d}{2}$ edges, and $\tilde{K}$ is an orthogonal projection of rank $d^2-1$.
\end{lemma}
\begin{proof}
Since $\tr(X_j)=\tr(Q_j-\frac{1}{d}I)=1-(d/d)=0$ for all $j=1,\ldots n$, the span of $\{X_j|j=1,\ldots n\}$ is contained in the subspace of traceless $d\times d$ matrices; thus ${\rm rk}(\tilde{K})\leq d^2-1$. Moreover,
$$
\tr(\tilde{K})=\sum_j
\left(|\langle \mathbf b_j, \mathbf b_j\rangle|^2-\frac{1}{d}\right)
=n \left(1-\frac{1}{d}\right)
$$
and
$$
\tr(\tilde{K}^2)=\sum_{j,k}\left(|\langle \mathbf b_j, \mathbf b_k\rangle|^2-\frac{1}{d}\right)^2 \\
\nonumber
= n (1-\frac{1}{d})^2 + 2|E|\frac{1}{d^2}
$$
where we have used that by (\ref{eq:<X_jX_k>}), the diagonal entries of $\tilde{K}$ are equal to $1-1/d$, the entries corresponding to orthogonal pairs are equal to $0-1/d = -1/d$ and the entries corresponding to unbiased pairs are equal to $1/d-1/d=0$. Taking into account ${\rm rk}(\tilde{K})\leq d^2-1$, the application of Lemma \ref{lemma:rk} to the Gram matrix $\tilde{K}$ gives

\begin{equation}\label{in2}
n^2(1-\frac{1}{d})^2\le (d^2-1)(n (1-\frac{1}{d})^2 + 2|E|\frac{1}{d^2}).
\end{equation}

After substituting $n=d(d+1)$ and rearranging, we get $|E|\geq (d+1)\binom{d}{2}$. This, together with
Corollary \ref{corr:ort_from_above}, proves that this inequality is actually an equality. Therefore, by the equality case of Lemma \ref{lemma:rk}, the matrix $\tilde{K}$ is a multiple of a projection; $\tilde{K}=\lambda P$ for some scalar $\lambda$ and orthogonal projection $P$. Also, the inequality in \eqref{in2} must also be an equality, which implies ${\rm rk}(P)={\rm rk}(\tilde{K})=d^2-1$. Therefore
$$
n(1-\frac{1}{d})=\tr(\tilde{K})=\tr(\lambda P)=\lambda (d^2-1),
$$
implying that $\lambda=1$ and hence that $\tilde{K}=P$.
\end{proof}

Consider the $n\times n$ matrix $A:=(d-1)I-d \tilde{K}$. By what we know
about the entries of $\tilde{K}$, it is easy to verify that
$$
A_{j,k}=\left\{\begin{matrix}
1, & {\rm if}\; \mathbf b_j \perp \mathbf b_k,
\\
0, & {\rm otherwise};
\end{matrix}\right.
$$
i.e.\! $A$ is simply the adjacency matrix of $G$. Thus, by having established that $\tilde{K}$ is a rank $d^2-1$ projection, we can precisely determine the spectrum of the adjacency matrix $A$, or, as it
is called in short, the spectrum of the graph $G$.

\medskip

In general, the spectrum of a graph does not determine its isomorphism class. That is, there exist graphs which are not isomorphic, yet have the same spectrum (including multiplicities); a curious fact that was first noted more than half a century ago \cite{spectra}. However, in this particular case, we can prove that $G$ must be a disjoint union of $(d+1)$ complete graphs, each with $d$ vertices.

\begin{theorem}\label{main}
Let $n=d(d+1)$ and $\mathbf b_1\ldots \vb_n\in \mathbb C^d$ be a collection of unit vectors such that $|\langle\mathbf b_j,\mathbf b_k\rangle|^2$ is either $0$ or $1/d$ for any $j\neq k$ (i.e.\! such that any two of them are either orthogonal or unbiased to each other). Then the vectors $\mathbf b_1 \ldots \mathbf b_n$ can be arranged into $d+1$ orthogonal bases, all being mutually unbiased to each other.
\end{theorem}

\begin{proof}
The eigenvalues of the matrix $A=(d-1)I-d \tilde{K}$, defined above, are $-1$ (with multiplicity $d^2-1$) and $d-1$ (with multiplicity $n-d^2+1=d+1$).
Let $\vone\in \Co^n$ denote the vector with entry 1 in each coordinate, and consider $h=\langle \vone, A\vone\rangle$. Due to the eigenvalues of $A$ we have $h\le (d-1)\langle \vone, \vone \rangle= (d-1)d(d+1),$ with equality only if $\vone$ is an eigenvector with eigenvalue $d-1$. Furthermore, $h$ is the sum of entries in $A$, which equals to twice the number of edges in $G$ (each edge being counted twice by the symmetry of $A$). Therefore $h=|E|=(d-1)d(d+1)$. This implies that $\vone$ is an eigenvector with eigenvalue $d-1$, which means that each vertex in $G$ has degree $d-1$ (in other words, the graph $G$ is $d-1$-regular). 

\medskip

It is also well-known that $\tr (A^3)$ equals to the number of (ordered) triangles present in $G$. By knowing the spectrum of $A$ we can calculate $\tr (A^3)=(-1)(d^2-1)+(d-1)^3(d+1)=(d^2-1)d(d-2)$. We claim that this implies that $G$ can be broken up to the disjoint union of $d+1$ complete graphs with $d$ vertices each. Indeed, the number of (ordered) triangles in a $d-1$ regular graph on $n$ vertices is at most $n(d-1)(d-2)$, because from each vertex we can choose $(d-1)(d-2)$ ordered pairs of edges, and the maximum number of triangles occurs if each of these pairs can be completed by a further edge to make a triangle. This happens if and only if $G$ breaks up to a disjoint union of $d+1$ complete graphs on $d$ vertices. 

\medskip 

In turn, this is equivalent to the vectors $\vb_1, \vb_2, \dots, \vb_n$ forming $d+1$ orthonormal bases, all being pairwise unbiased with respect to each other. 
\end{proof}

\end{document}